\documentclass{article}
\usepackage{algorithm}
\usepackage[noend]{algorithmic}
\usepackage{amsmath, amssymb}
\usepackage{qtree}
\usepackage{graphicx}
\usepackage{url}
\usepackage{comment}
\usepackage{amsthm}

\title{A Versatile Algorithm to Generate Various Combinatorial Structures}
\author{Pramod Ganapathi\thanks{Associate Software Engineer, IBM, India; pramodghegde@in.ibm.com}, Rama B\thanks{Graduate Student, Department of Computer Science and Automation, Indian Institute of Science, India; ramab@csa.iisc.ernet.in}}

\newtheorem{defn}{Definition}
\newtheorem{theorem}[defn]{Theorem}
\newtheorem{lemma}[defn]{Lemma}

\newtheorem{corollary}[defn]{Corollary}

\begin{document}

\thispagestyle{empty}

\label{firstpage}
\maketitle

\begin{abstract}
Algorithms to generate various combinatorial structures find tremendous importance in computer science. In this paper, we begin by reviewing an algorithm proposed by Rohl~\cite{rohl1} that generates all unique permutations of a list of elements which possibly contains repetitions, taking some or all of the elements at a time, in any imposed order. The algorithm uses an auxiliary array that maintains the number of occurrences of each unique element in the input list. We provide a proof of correctness of the algorithm. We then show how one can efficiently generate other combinatorial structures like combinations, subsets, $n$-Parenthesizations, derangements and integer partitions \& compositions with minor changes to the same algorithm.
\end{abstract}

\floatname{algorithm}{Procedure}

\section{Introduction}
\label{sec:intro}

Algorithms which generate combinatorial structures like permutations, combinations, etc. come under the broad category of combinatorial algorithms~\cite{book_KrSt}. These algorithms find importance in areas like Cryptography, Molecular Biology, Optimization, Graph Theory, etc. Combinatorial algorithms have had a long and distinguished history. Surveys by Akl~\cite{selim_combsurvey}, Sedgewick~\cite{sedgewick_perm}, Ord Smith~\cite{ordsmith_1,ordsmith_2}, Lehmer~\cite{lehmer_tricks}, and of course Knuth's thorough treatment~\cite{knuth_perm,knuth_fas3}, successfully capture the most important developments over the years. The literature is clearly too large for us to do complete justice. We shall instead focus on those which have directly impacted our work.

The well known algorithms which solve the problem of generating all $n!$ permutations of a list of $n$ unique elements are Heap Permute~\cite{heap_permute}, Ives~\cite{ives_perm} and Johnson-Trotter~\cite{trotter_perm}. In addition, interesting variations to the problem are - algorithms which generate only unique permutations of a list of elements with possibly repeated elements, like \cite{bratley_unique, chase_unique, sag_repetitions}, algorithms which generate permutations in lexicographic order, like \cite{smith_lexico1, schrack_lexico, shen_lexico} and algorithms which genenerate permutations taking $r (< n)$ elements at a time. With regard to these variations, the most complete algorithm, in our opinion, is the one proposed by Rohl~\cite{rohl1} which possesses the following salient features:
\begin{enumerate}
\item Generates all permutations of a list considering some or all elements at a time
\item Generates only unique permutations
\item Generates permutations in any imposed order
\end{enumerate}
It is of course possible to extend any permutation generating algorithm to exhibit all of the above features. However, when such naive algorithms are given inputs similar to \emph{aaaaaaaaaaaaaaaaaaab}\footnote{There are 20 unique permutations of length 20}, the number of wasteful computations can be unmanageably huge. Clearly, the need for specially tailored algorithms is justified.

We provide a rigorous proof of correctness of Rohl's algorithm. We also show interesting adaptations of the algorithm to generate various other combinatorial structures efficiently.

The paper is organised as follows: we begin by introducing terminologies and notations that will be used in the paper in Section~\ref{sec:prelim}, and presenting Rohl's algorithm in Section~\ref{sec:algo}. Section~\ref{sec:anal} contains a detailed analysis of the algorithm, which also includes the proof of correctness of the algorithm (missing in Rohl's paper). We then present various efficient extensions of Rohl's algorithm in Section~\ref{sec:exten}, to generate many more combinatorial structures.

Rohl's algorithm involves two stages. The first stage builds up an auxiliary array that maintains the number of occurrences of each unique element in the input list. The second stage uses this auxiliary array to effect and generates the required combinatorial structure. We present a recursive representation of the algorithm. It would be quite easy to transform the recursive version to an iterative one~\cite{rec2iter}.

\section{Preliminaries}
\label{sec:prelim}

In this section, we introduce the key terms, concepts and notations that will be used in the paper.

\subsection{Basic Definitions} 

\noindent \textbf{Set:} A collection of elements without repetitions

\noindent \textbf{List:} A collection of elements possibly with repetitions

\noindent \textbf{Permutation:} A permutation of a list is an arrangement of the elements of the list in any order, taking some or all elements at a time \\
E.g: Consider a list $\{a,b,c,d\}$; Various permutations are - $\{d,a,b,c\}$, $\{a,d,c\}$, $\{b,d,a\}$, $\{a,d\}$, $\{d,c\}$, $\{b,a\}$

\noindent \textbf{$r$-Permutation Set:} The $r$-permutation set of a list is defined as the set of all possible $r$-permutations of the list, where an $r$-permutation is defined as a permutation taking exactly $r$ $(1 \le r \le size(list))$ elements at a time\\
E.g: Consider a list $\{a, b, c\}$. We get: $\left\{a, b, c\right\}$, $\left\{ab, ac, ba, bc, ca, cb\right\}$ and $\{abc$, $acb$, $bac$, $bca$, $cab$, $cba\}$\footnote{Each permutation in an $r$-permutation set is actually a list of elements. For compactness, we represent the elements in concatenated form.} as the $1$-permutation set, $2$-permutation set and $3$-permutation set, respectively.

The algorithm takes as input:
\begin{itemize}
 \item A list of $n$ elements to permute, called input list: $\mathcal{L} = \{l_1, l_2, \ldots, l_n\}$ containing $p (\le n)$ unique elements
 \item An integer $r$, $1 \le r \le n$, which denotes the number of elements to be taken at a time during permuting
 \item A set of $p$ elements which are actually the $p$ unique elements of $\mathcal{L}$ in a particular sequence, called Order Set (or just Order): $\mathcal{O} = \{o_1, o_2, \ldots, o_p\}$
\end{itemize}

It produces as output the $r$-permutation set of $\mathcal{L}$, which is a set of $m$ $r$-permutations: $\mathcal{P}^{r}(\mathcal{L}) = \{\mathcal{P}^r_1, \mathcal{P}^r_2, \ldots, \mathcal{P}^r_m\}$ where each $\mathcal{P}^r_i$ is in turn is a list of $r$ elements: $\mathcal{P}^r_i = \{p^r_{i1}, \ldots, p^r_{ir}\}$ $\forall\,\,(1 \le i \le m$). The algorithm works such that $P^r(\mathcal{L})$ follows the order $\mathcal{O}$ (definition and explanation in Section~\ref{sec:order}). Quite obviously, $p^r_{ij} \in \mathcal{L}\,\,\forall\,(1 \le i \le m),\,(1 \le j \le r)$. Note that we use the terms $r$-permutation ({set}) and permutation ({set}) interchangeably, if clear from the context.

\subsection{Order of Permutation Set}
\label{sec:order}

Consider the $3$-permutation set of the list $\{a,b,c\}: \{abc$, $acb$, $bac$, $bca$, $cab$, $cba\}$. We see that the permutations are in lexicographic order, i.e they follow the order $\{a,b,c\}$. At the same time, the $3$-permutation set \{$bac$, $bca$, $abc$, $acb$, $cba$, $cab$\} follows the order $\{b, a, c\}$. We shall now introduce a formal definition of what we mean when we say  $P^r(\mathcal{L})$ follows the order $\mathcal{O}$.

\noindent \textbf{Discriminating Index:} The discriminating index between two different $r$-permutations is the first (minimum) index in both permutations at which the elements differ from each other. Given two permutations $\mathcal{P}^r_a$ and $\mathcal{P}^r_b$, we denote the discriminating index as 
\[\mathcal{D}(\mathcal{P}^r_a, \mathcal{P}^r_b) = \,\,\textbf{min}\{x: p^r_{ax} \neq p^r_{bx}\}\]
Also, assume that $\mathcal{I}_{\mathcal{O}}(p^r_{ij})$ denotes the index of the element $p^r_{ij}$ in $\mathcal{O}$. Thus $\mathcal{I}_{\mathcal{O}}(p^r_{ij}) = x \,\,\Longleftrightarrow\,\, o_x = p^r_{ij}$.

$\mathcal{P}^r(\mathcal{L})$ is said to \emph{follow} an order $\mathcal{O}$ if:
\begin{itemize}
\item $\forall\,(1 \le i \le m),\,(1 \le j \le r) \qquad p^r_{ij} \in \mathcal{O}$
\item  For every pair of permutations in $\mathcal{P}^r(\mathcal{L})$, the elements at the discriminating index have the same positions relative to each other in $\mathcal{O}$, as do the pair of permutations in $\mathcal{P}^r(\mathcal{L})$ \\
$\forall \,\, (1 \le a < b \le m): \,\,\,\,\,\,\,\mathcal{I}_\mathcal{O}(p^r_{ad}) < \mathcal{I}_\mathcal{O}(p^r_{bd})$
\end{itemize}

where $d = \mathcal{D}(\mathcal{P}^r_a, \mathcal{P}^r_b)$. It has to be noted that similar definitions hold for both sets of combinations and derangements too.

\subsection{Auxiliary Array}
\label{hashorder}

Central to the algorithm is an auxiliary array denoted as $CountArray$. It is an array of integers built in such a way that $CountArray$[$i$] gives the number of occurrences of element $o_i$ in the {input list $\mathcal{L}$}. $CountArray$ is of size $p$ as it is maintained parallel to $\mathcal{O}$.

As an example, consider $\mathcal{L} = \{a,c,b,a,a,c\}$, $\mathcal{O} = \{c,a,b\}$; here $n = 6$ and $p = 3$. We have $o_1$ (= c) occurring twice, $o_2$ (= a) occurring thrice and $o_3$ (= b) occurring once. Thus we must have $CountArray = \{2,3,1\}$. Instead, if we had $\mathcal{O} = \{a,b,c\}$, we would have $CountArray = \{3,1,2\}$.

\section{The Algorithm}
\label{sec:algo}

In this section, we introduce the algorithm. The algorithm takes $\mathcal{L}$, $r$ and $\mathcal{O}$ as input and produces $\mathcal{P}^r(\mathcal{L})$ as output, while ensuring that $\mathcal{P}^r(\mathcal{L})$ follows $\mathcal{O}$. The first stage involves building $CountArray$ while the second stage involves recursive generation of $\mathcal{P}^r(\mathcal{L})$ using $CountArray$.

\subsection{Building $CountArray$}

The first stage involves building $CountArray$ parallel to $\mathcal{O}$ in a way that $CountArray[i]$ is an integer which represents the number of occurrences of $o_i$ in $\mathcal{L}$. Procedure~\ref{countarray} is a pseudocode representation of how to build $CountArray$, given $\mathcal{L}$ and $\mathcal{O}$. We assume the existence of the function $\mathcal{I}_{\mathcal{O}}$ which takes an element of $\mathcal{L}$ as a parameter and returns the position of the element in $\mathcal{O}$. We can say ${\mathcal{I}_{\mathcal{O}}(l_i)} = x \Leftrightarrow o_x = l_i$. A naive implementation of $\mathcal{I}_{\mathcal{O}}$ would be to perform a linear search over $\mathcal{O}$, and this would have a worst case runtime of $O(p)$. However, if needed, we can achieve $O(1)$ runtime using hash functions. We leave the implementation details to the reader.

\begin{algorithm}
\caption{: Building $CountArray$}
\label{countarray}
\begin{algorithmic}[1]
 \STATE \textbf{Input:} $\mathcal{L}$ and $\mathcal{O}$
 \STATE \textbf{Output:} $CountArray$
 \medskip
 \STATE $CountArray[1 \ldots p] \gets 0$
 \FOR{$i$ $\gets 1$ to $n$}
	\STATE $pos \gets \mathcal{I}_{\mathcal{O}}(l_i)$
	\STATE $CountArray[pos] \gets CountArray[pos] + 1$
 \ENDFOR
\end{algorithmic}
\end{algorithm}

\subsection{Generating Permutations}

Procedure~\ref{countarray} builds $CountArray$. The second stage involves using the recursive routine APR to generate $\mathcal{P}^r(\mathcal{L})$, following order $\mathcal{O}$, as output. To this end, we use an output array $\mathcal{R}$.

Procedure~\ref{permute} is a pseudocode representation of APR. It assumes non-local existence of the following: $CountArray$, $\mathcal{O}$, $r$, $p$ and $\mathcal{R}$. Also, it has a local parameter - an integer $index$. Thus the function's header takes the form \textbf{APR($index$)}. It is invoked by the function call APR(1). As APR recursively calls itself ($1 \le index \le r$), $\mathcal{R}[1, \ldots, r]$ is populated using $CountArray$.

\begin{algorithm}
\begin{algorithmic}[1]
\caption{: APR($index$) - Permutations}
\label{permute}
\STATE \textbf{Local:} $index$
\STATE \textbf{Global:} $CountArray$, $\mathcal{O}$, $r$, $p$ and $\mathcal{R}$
\medskip
\IF{$index > r$}
	\STATE Print $\mathcal{R}[1, \ldots, r]$
	\RETURN
\ELSE
	\FOR{$i \gets 1$ to $p$}
		\IF{$CountArray[i] \ge 1$} \label{proc:apr:line:cond}
			\STATE $\mathcal{R}[index] \gets o_i$ 
			\STATE $CountArray[i] \gets CountArray[i] - 1$
			\STATE APR($index + 1$)
			\STATE $CountArray[i] \gets CountArray[i] + 1$
		\ENDIF
	\ENDFOR
\ENDIF
\end{algorithmic}
\end{algorithm}

\section{Algorithm Analysis}
\label{sec:anal}

In this section, we begin by explaining how APR works and showing the recursion tree for a particular example. We then give a formal proof of its correctness. We conclude by establishing an upper bound on the runtime of APR.

We know that the contents of $CountArray$ are constantly changing. Given a particular index in $\mathcal{R}$ we assign to it an element $o_i \in \mathcal{O}\,(1 \le i \le p)$ only if $CountArray[i] \ge 1$, at that particular time. We call the assignment \emph{minimal} if we choose the $o_i$ corresponding to the minimum $i$ for which $CountArray[i] \ge 1$. More formally, we can say that an assignment of $o_i$ to a particular index in $\mathcal{R}$ is minimal iff $\nexists\,\,o_x \in \mathcal{O};\,1 \le x < i$ such that $CountArray[x] \ge 1$. Also we say \emph{populates $\mathcal{R}[i \ldots j]$ minimally} to signify that minimal assignment is done at every index from $i$ to $j$ (inclusive).

\subsection{Recursion Tree}
\label{subsec:rectree}

Let us analyse how APR functions by referring to Procedure~\ref{permute}. The function, at any recursion level, begins by searching for the first available element, i.e. the first index $i$ at which $CountArray[i] \ge 1$ (lines 7-8). Once found, the corresponding element ($o_i$) is assigned at the current index in $\mathcal{R}$ (line 9). It then decrements $CountArray[i]$ to reflect the fact that that particular element was assigned (line 10). After that, there is a recursive call and the control moves to the next recursion level and deeper thereafter performing the same set of functions (lines 7-10). Once the control returns to the current recursion level, it de-assigns the element that was assigned to the current index in $\mathcal{R}$, by incrementing $CountArray[i]$ (line 10). It then moves to the next available element in $CountArray$ and performs the process of assigning and recursion all over again.

\begin{figure*}[!htp]
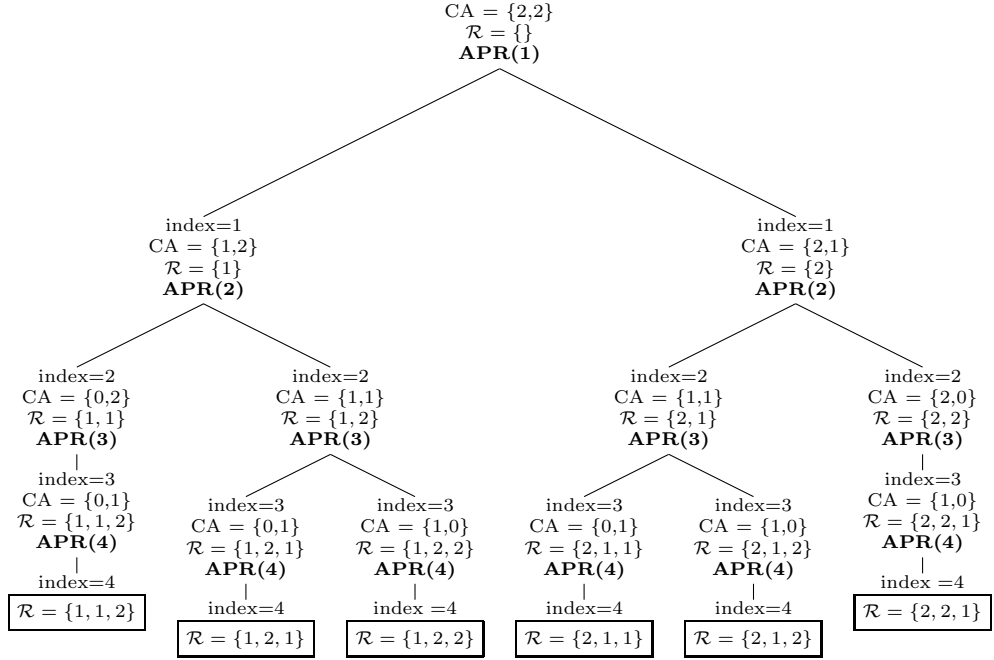
 \centering
{\scriptsize \Tree [.{CA = \{2,2\} \\ $\mathcal{R}=\{\}$ \\ \textbf{APR(1)}} 
	[.{index=1 \\ CA = \{1,2\} \\ $\mathcal{R}=\{1\}$ \\ \textbf{APR(2)}} 
		[.{index=2 \\ CA = \{0,2\} \\ $\mathcal{R}=\{1,1\}$ \\ \textbf{APR(3)}}
			 [.{index=3 \\ CA = \{0,1\} \\ $\mathcal{R}=\{1,1,2\}$ \\ \textbf{APR(4)}} 
				[.{index=4 \\ \fbox{$\mathcal{R}=\{1,1,2\}$}} ] ] ]
		 [.{index=2 \\ CA = \{1,1\} \\ $\mathcal{R}=\{1,2\}$ \\ \textbf{APR(3)}}
			 [.{index=3 \\ CA = \{0,1\} \\ $\mathcal{R}=\{1,2,1\}$ \\ \textbf{APR(4)}}
				 [.{index=4 \\ \fbox{$\mathcal{R}=\{1,2,1\}$}} ] ] 
			[.{index=3 \\ CA = \{1,0\} \\ $\mathcal{R}=\{1,2,2\}$ \\ \textbf{APR(4)}} 
				[.{index =4 \\ \fbox{$\mathcal{R}=\{1,2,2\}$}} ] ] ] ]
	 [.{index=1 \\ CA = \{2,1\} \\ $\mathcal{R}=\{2\}$ \\ \textbf{APR(2)}}
		 [.{index=2 \\ CA = \{1,1\} \\ $\mathcal{R}=\{2,1\}$ \\ \textbf{APR(3)}}
			 [.{index=3 \\ CA = \{0,1\} \\ $\mathcal{R}=\{2,1,1\}$ \\ \textbf{APR(4)}}
				 [.{index=4 \\ \fbox{$\mathcal{R}=\{2,1,1\}$}} ] ]
			 [.{index=3 \\ CA = \{1,0\} \\ $\mathcal{R}=\{2,1,2\}$ \\ \textbf{APR(4)}} 
				[.{index=4 \\ \fbox{$\mathcal{R}=\{2,1,2\}$}} ] ] ]
		 [.{index=2 \\ CA = \{2,0\} \\ $\mathcal{R}=\{2,2\}$ \\ \textbf{APR(3)}}
			 [.{index=3 \\ CA = \{1,0\} \\ $\mathcal{R}=\{2,2,1\}$ \\ \textbf{APR(4)}} 
				[.{index =4 \\ \fbox{$\mathcal{R}=\{2,2,1\}$}} ] ] ] ] ]}
\caption{Recursion Tree of APR for $\mathcal{L}=\{1,1,2,2\}$; $r=3$; $\mathcal{O} = \{1,2\}$}
\label{rectree}
\end{figure*}

 We venture an example to better illustrate the process. Assume we have $\mathcal{L} = \{1,1,2,2\}$, $r = 3$ and $\mathcal{O}=\{1,2\}$. When $\mathcal{L}$ and $\mathcal{O}$ are fed to Procedure~\ref{countarray}, we would obtain $CountArray$[1] = $CountArray$[2] = 2. Procedure~\ref{permute} is then invoked. The recursion tree that would be obtained is shown in Figure~\ref{rectree} ($CountArray$ is abbreviated to CA). The root of the tree represents the initial state, i.e. $CountArray = \{2,2\}$ and $\mathcal{R}=\{\}$. Every descendant of the root represents a particular recursion depth, indicated by the value of $index$. At every node, we indicate the values in $CountArray$ and $\mathcal{R}$ just before going down to the next recursion level followed by the function call that initiates the next recursion level. Once we go four levels deep into the recursion, i.e. $index=4$, the contents of $\mathcal{R}$ would be output because $index > r$ (lines 3-4 of Procedure~\ref{permute}). The control then returns to the previous recursion level (line 5). We observe that $\mathcal{P}^3(\mathcal{L})$, i.e. the collection of all leaves of the recursion tree from left to right, follows the order $\{1, 2\}$.

\subsection{Proof of Correctness}
\label{subsec:proof}

In this section, we shall use $\mathcal{P}_i$ to denote a general $r$-permutation instead of $\mathcal{P}^r_i$. Also, we shall use $\mathcal{P}_i[x]$, instead of $p_{ix}$, to denote the $x^{th}$ element of the $i^{th}$ permutation.

\begin{theorem}
 The algorithm generates only valid permutations of $\mathcal{L}$.
\end{theorem}

\begin{proof}
 We assign an element $o_i$ to any particular index in $\mathcal{R}$ only if we have $CountArray[i] \ge 1$. Also, we decrement or increment $CountArray[i]$ whenever $o_i$ is assigned or de-assigned, respectively. This process ensures that APR only assigns elements that were actually present in $\mathcal{L}$ and no element is assigned more times than it occurs. Hence we can say that every permutation that is output is valid.
\end{proof}

\begin{theorem}
\label{thm:unique_order}
The algorithm generates only unique permutations of $\mathcal{L}$ following order $\mathcal{O}$.
\end{theorem}

Before we prove Theorem~\ref{thm:unique_order}, we must understand how the algorithm generates a permutation subsequent to an already generated one. Once a permutation is output, the program searches backwards from the end of $\mathcal{R}$ for an index where it can assign, from the available elements, an element that has a \emph{higher} position in $\mathcal{O}$ than the previously assigned element. This means that, if $o_x$ were assigned at a particular index, the algorithm checks if it can assign any among $o_{x+1} \ldots o_p$ at the same index. At the \emph{first instance} (referred to as Property 1) of such an index, the algorithm assigns the \emph{first among} (referred to as Property 2) $o_{x+1} \ldots o_p$, whichever is available. Note that this index would be the discriminating index between the previously output permutation and the next permutation that is going to be output. Once this is done, the program \emph{populates the rest of $\mathcal{R}$ minimally} (referred to as Property 3).

\begin{proof}[Proof of Theorem~\ref{thm:unique_order}]
This process ensures that, between two consecutive permutations, we will have at least one index where the elements differ. Given that at this index, i.e. discriminating index, we assign one of $o_{x+1} \ldots o_p$, we can be sure that the permutation generated subsequent to a particular permutation will occupy a higher position in $\mathcal{P}^r(\mathcal{L})$ (when it follows order $\mathcal{O}$). 
\end{proof}

\begin{lemma}
\label{lemma:disc}
 Given three consecutive permutations $\mathcal{P}_{i}$, $\mathcal{P}_j$ and $\mathcal{P}_{k}$, we have: \\ $\mathcal{D}(\mathcal{P}_{i},\mathcal{P}_j) \ge \mathcal{D}(\mathcal{P}_{i},\mathcal{P}_{k})$.
\end{lemma}

\begin{proof}
 This shall be proved by contradiction. Let us assume we have $\mathcal{P}_{i}$, $\mathcal{P}_j$, $\mathcal{P}_{j}$ such that
\begin{align}
d_{ij} &< d_{ik} \label{eqn:contrad1}
\end{align}
where $d_{ij}, d_{jk}, d_{ik}$ represent each of $\mathcal{D}(\mathcal{P}_{i},\mathcal{P}_j)$, $\mathcal{D}(\mathcal{P}_{j},\mathcal{P}_{k})$, $\mathcal{D}(\mathcal{P}_{i},\mathcal{P}_{k})$, respectively. By the definition of discriminating index in Section~\ref{sec:order}, we have the following equations.
\begin{align}
\mathcal{P}_{i}[1 \ldots (d_{ij} - 1)] &= \mathcal{P}_{j}[1 \ldots (d_{ij} - 1)] \label{eqn:eqp1} \\
\mathcal{I}_{\mathcal{O}}(\mathcal{P}_{i}[d_{ij}]) &< \mathcal{I}_{\mathcal{O}}(\mathcal{P}_j[d_{ij}]) \label{eqn:disc1}
\end{align}
and
\begin{align}
\mathcal{P}_{i}[1 \ldots (d_{ik} - 1)] &= \mathcal{P}_{k}[1 \ldots (d_{ik} - 1)] \label{eqn:eqp2} \\
\mathcal{I}_{\mathcal{O}}(\mathcal{P}_{i}[d_{ik}]) &< \mathcal{I}_{\mathcal{O}}(\mathcal{P}_{k}[d_{ik}])
\end{align}
From equations \ref{eqn:contrad1} and \ref{eqn:eqp2},  we can say $\mathcal{P}_{k}[d_{ij}] = \mathcal{P}_{i}[d_{ij}]$. By using this in equation~\ref{eqn:disc1}, we get
\begin{align}
\mathcal{I}_{\mathcal{O}}(\mathcal{P}_{k}[d_{ij}]) &< \mathcal{I}_{\mathcal{O}}(\mathcal{P}_j[d_{ij}]) \label{eqn:finp1}
\end{align}
Also, from equations \ref{eqn:eqp1} and \ref{eqn:eqp2} and equation \ref{eqn:contrad1}, we get
\begin{align}
\mathcal{P}_{k}[1 \ldots (d_{ij} - 1)] &= \mathcal{P}_{j}[1 \ldots (d_{ij} - 1)]  \label{eqn:finp2}
\end{align}
Equations \ref{eqn:finp1} and \ref{eqn:finp2} imply that $\mathcal{P}_{k}$ comes before $\mathcal{P}_{j}$, when order of consideration is $\mathcal{O}$, which is a contradiction to the order of the permutations. This in turn implies that the initial assumption is wrong. 
\end{proof}

\begin{corollary}
\label{corr:disc}
$\mathcal{D}(\mathcal{P}_{j},\mathcal{P}_k) \ge \mathcal{D}(\mathcal{P}_{i},\mathcal{P}_{k})$.
\end{corollary}

\begin{proof}
 Equations \ref{eqn:eqp1} and \ref{eqn:eqp2} follow from the definition of discriminating index. From Lemma~\ref{lemma:disc} we have $d_{ij} \ge d_{ik}$, thus we get
\begin{align}
\mathcal{P}_{k}[1 \ldots (d_{ik} - 1)] &= \mathcal{P}_{j}[1 \ldots (d_{ik} - 1)] 
\end{align}
Thus the corollary follows. 
\end{proof}

\begin{lemma}
\label{lemma:nextperm}
 Given two successive permutations that APR generates, say $\mathcal{P}_a$ and $\mathcal{P}_b$, $\nexists\, \mathcal{P}_x \in \left(\mathcal{P}^r(\mathcal{L})/ \{\mathcal{P}_a,\mathcal{P}_b\}\right)$: such that $\mathcal{P}_x$ occurs between $\mathcal{P}_a$ and $\mathcal{P}_b$.\footnote{$\mathcal{A}/\{A_1,A_2\}$ Denotes $\mathcal{A}$ excluding $A_1$ and $A_2$}
\end{lemma}

\begin{proof}
This shall be proved by contradiction. Assume there exists such a $\mathcal{P}_x$. Let us take $d_{ab} = \mathcal{D}(\mathcal{P}_a,\mathcal{P}_b)$, $d_{ax} = \mathcal{D}(\mathcal{P}_a,\mathcal{P}_x)$ and $d_{xb} = \mathcal{D}(\mathcal{P}_x,\mathcal{P}_b)$. By the analysis of how APR generates a permutation subsequent to an already generated one
\begin{align}
 d_{ax} \ngtr d_{ab} \qquad &\text{[From Property 1]} \\
 d_{xb} \ngtr d_{ab} \qquad &\text{[From Property 3]}
\end{align}
 Also we have $d_{ax} \nless d_{ab}$ from Lemma~\ref{lemma:disc} and $d_{xb} \nless d_{ab}$ from Corollary~\ref{corr:disc}. Thus $d_{ax} = d_{bx} = d_{ab} (= \lambda\text{ say})$ is the only possibility.

For $\mathcal{P}_x$ to be between $\mathcal{P}_a$ and $\mathcal{P}_b$: $\mathcal{I}_{\mathcal{O}}(\mathcal{P}_a[\lambda]) < \mathcal{I}_{\mathcal{O}}(\mathcal{P}_x[\lambda]) < \mathcal{I}_{\mathcal{O}}(\mathcal{P}_b[\lambda])$. This is a direct contradiction to Property 2. Thus we can have no $\mathcal{P}_x$ in between $\mathcal{P}_a$ and $\mathcal{P}_b$. 
\end{proof}

\begin{theorem}
 The algorithm generates all $r$-permutations of $\mathcal{L}$.
\end{theorem}

\begin{proof}
 We begin by proving that APR correctly generates the first permutation. By Lemma~\ref{lemma:nextperm}, we have proved that the process of generating the next permutation given the current one is accurate in that the immediate next permutation, in accordance with the order $\mathcal{O}$, is generated. We then prove that APR terminates once all permutations are generated.

At the beginning of the algorithm, we perform a minimal assignment at the first index, i.e. $index=1$. The program then moves down a recursion level and once again performs a minimal assignment. This process continues until the last recursion level. In the $(r+1)^{th}$ recursion level, the program outputs the permutation. Clearly the first permutation that would be output by APR would be formed by populating $\mathcal{R}[1 \ldots r]$ minimally.

The program control returns from a level of recursion only after looping in $[1, p]$ is complete, i.e. when there are no more available elements that can be allotted at that corresponding index in $\mathcal{R}$. We know that the range $[1, p]$ is finite because the size of $\mathcal{L}$ is finite. Also the maximum recursion depth of the program is finite because $r$ is finite. This would ensure that every recursion level would eventually end which in turn ensures that the program will eventually terminate. 
\end{proof}

\subsection{Running Times}
\label{sec:anal:subsec:runtime}

To establish an upper bound on the running time of APR, let us analyse the recursion tree shown in Figure~\ref{rectree}. We see that the number of leaves are exactly equal to the number of unique $r$-permutations of $\mathcal{L}$, i.e. $|\mathcal{P}^r(\mathcal{L})|$. Each of the leaves are exactly $(r+1)$ levels below the root. The critical operations, i.e. looping and searching for available elements, are done on the first $r$ levels.\footnote{In the $(r+1)^{th}$ level, the permutations are output.} At each level, we loop in the range $[1, p]$. Thus we can say that the running time of APR is bounded by $O\left(|\mathcal{P}^r(\mathcal{L})| \times r \times p\right)$ in the worst case.

It is to be noted that the runtime of Procedure~\ref{countarray} has been ignored in comparison to the runtime of Procedure~\ref{permute}.

\section{Extensions}
\label{sec:exten}

In this section we show how the smallest of changes to Procedure~\ref{permute} helps us generate many more combinatorial structures. Please note that the proof of correctness of each of the following algorithms follows trivially from the rigorous analysis in Section~\ref{subsec:proof}.

\subsection{Derangements}

A derangement of a list is a permutation of the list in which none of the elements appear in their original positions. Some previous work on derangement generating algorithms can be found in \cite{deran2, deran1}.

With just two changes to Procedure~\ref{permute}, we are able to generate derangements - 
\begin{itemize}
 \item Non-local existence of $\mathcal{L}$, in addition to the other entities, is required
 \item Before assigning an element at a particular index, we perform an additional check to ensure that the same element does not occur at the exact same index in $\mathcal{L}$, i.e., if we were to assign $o_i$ to $\mathcal{R}[index]$, we would need to ensure that $l_{index} \neq o_i$
\end{itemize}

The procedure shall be referred to as APR2 (shown in Procedure~\ref{proc:apr2}). It is invoked by the call APR2(1). All we do is perform an additional check to ensure that the element being assigned at the current position does not appear in $\mathcal{L}$ at the same position (Line~\ref{proc:apr2:line:cond}). $CountArray$ is built exactly as was illustrated in Procedure~\ref{countarray}.

Extending the algorithm to generate partial derangements, like in \cite{partialderan}, is easily done. Figure~\ref{fig:rectreederan} shows the recursion tree when we generate derangements of $\mathcal{L} = \{1,2,3\}$; $r=3$ and $\mathcal{O}=\{1,2,3\}$. It is to be noted that the program can output $r$-derangements, as well all derangements following an imposed order.

\begin{algorithm}
\begin{algorithmic}[1]
\caption{: APR2($index$) - Derangements}
\label{proc:apr2}
\STATE \textbf{Local:} $index$
\STATE \textbf{Global:} $CountArray$, $\mathcal{O}$, $r$, $p$ and $\mathcal{R}$
\medskip
\IF{$index > r$}
	\STATE Print $\mathcal{R}[1, \ldots, r]$
	\RETURN
\ELSE
	\FOR{$i \gets 1$ to $p$}
		\IF{$CountArray[i] \ge 1$ \textbf{and} $l_{index} \neq o_i$} \label{proc:apr2:line:cond}
			\STATE $\mathcal{R}[index] \gets o_i$ 
			\STATE $CountArray[i] \gets CountArray[i] - 1$
			\STATE APR2($index + 1$)
			\STATE $CountArray[i] \gets CountArray[i] + 1$
		\ENDIF
	\ENDFOR
\ENDIF
\end{algorithmic}
\end{algorithm}

\begin{figure}[!htp]
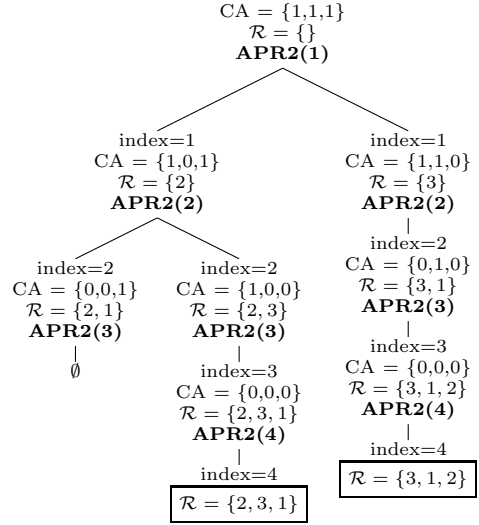
 \centering
{\scriptsize \Tree [.{CA = \{1,1,1\} \\ $\mathcal{R}=\{\}$ \\ \textbf{APR2(1)}} 
	[.{index=1 \\ CA = \{1,0,1\} \\ $\mathcal{R}=\{2\}$ \\ \textbf{APR2(2)}} 
		[.{index=2 \\ CA = \{0,0,1\} \\ $\mathcal{R}=\{2,1\}$ \\ \textbf{APR2(3)}}
			[.{$\emptyset$} ] ]
		[.{index=2 \\ CA = \{1,0,0\} \\ $\mathcal{R}=\{2,3\}$ \\ \textbf{APR2(3)}}
			[.{index=3 \\ CA = \{0,0,0\} \\ $\mathcal{R}=\{2,3,1\}$ \\ \textbf{APR2(4)}}
				 [.{index=4 \\ \fbox{$\mathcal{R}=\{2,3,1\}$}} ] ] ] ]
	 [.{index=1 \\ CA = \{1,1,0\} \\ $\mathcal{R}=\{3\}$ \\ \textbf{APR2(2)}}
		 [.{index=2 \\ CA = \{0,1,0\} \\ $\mathcal{R}=\{3,1\}$ \\ \textbf{APR2(3)}}
			 [.{index=3 \\ CA = \{0,0,0\} \\ $\mathcal{R}=\{3,1,2\}$ \\ \textbf{APR2(4)}}
				 [.{index=4 \\ \fbox{$\mathcal{R}=\{3,1,2\}$}} ] ] ] ] ]}
\caption{Recursion Tree of APR2: $\mathcal{L} = \{1,2,3\}$; $r=3$ and $\mathcal{O}=\{1,2,3\}$}
\label{fig:rectreederan}
\end{figure}

\subsection{Combinations}
\label{subsec:comb}

\begin{figure*}[!htp]
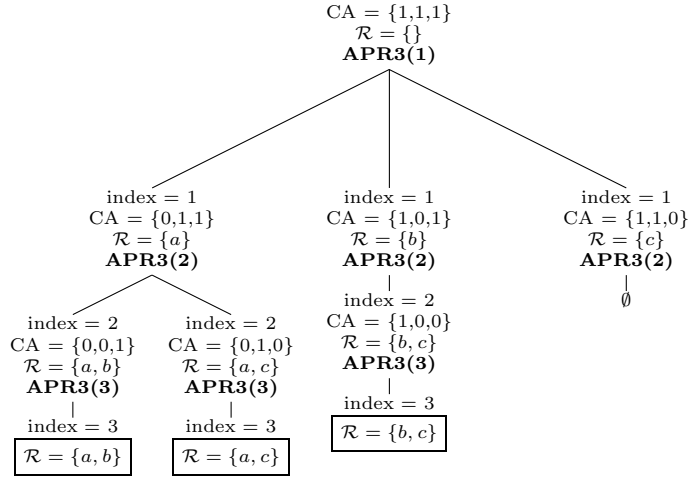
 \centering
{\scriptsize \Tree [.{CA = \{1,1,1\} \\ $\mathcal{R}=\{\}$ \\ \textbf{APR3(1)}} 
	[.{index = 1 \\ CA = \{0,1,1\} \\ $\mathcal{R}=\{a\}$ \\ \textbf{APR3(2)}} 
		[.{index = 2 \\ CA = \{0,0,1\} \\ $\mathcal{R}=\{a,b\}$ \\ \textbf{APR3(3)}}
			 [.{index = 3 \\ \fbox{$\mathcal{R}=\{a,b\}$}} ] ] 
		[.{index = 2 \\ CA = \{0,1,0\} \\ $\mathcal{R}=\{a,c\}$ \\ \textbf{APR3(3)}}
			[.{index = 3 \\ \fbox{$\mathcal{R}=\{a,c\}$}} ] ] ]
	[.{index = 1 \\ CA = \{1,0,1\} \\ $\mathcal{R}=\{b\}$ \\ \textbf{APR3(2)}} 
		[.{index = 2 \\ CA = \{1,0,0\} \\ $\mathcal{R}=\{b,c\}$ \\ \textbf{APR3(3)}}
			[.{index = 3 \\ \fbox{$\mathcal{R}=\{b,c\}$}} ] ] ]
	[.{index = 1 \\ CA = \{1,1,0\} \\ $\mathcal{R}=\{c\}$ \\ \textbf{APR3(2)}}
		 [.{$\emptyset$} ] ] ]}
\caption{Recursion Tree of APR3 for $\mathcal{L}=\{a,b,c\}$; $r = 2$}
\label{fig:rectreecomb}
\end{figure*}

A combination of a list of elements is a selection of some or all of the elements. It is a selection where the sequence of elements in the selection is not important. Some combination generating algorithms that have been proposed are Chase~\cite{chase_comb}, Ehrlich~\cite{gideon_perm, gideon_PC} and Lam-Soicher~\cite{clement_Cminchange}. 

We would require to make just one change to Procedure~\ref{permute} to be able to produce all possible $r$-combinations\footnote{Definition analogous to definition of $r$-permutation.} of an input list. The new procedure shall be referred to as APR3 (shown in Procedure~\ref{proc:apr3}). All we do is - before assigning an element at a particular index in $\mathcal{R}$, we make sure that the element occupies a position in $\mathcal{O}$ that is greater than the position of element at the previous index, in $\mathcal{O}$. At the first index however, we may assign any element. The change is in Line~\ref{proc:apr3:line:cond} of Procedure~\ref{proc:apr3}. $CountArray$ is built in exactly as was illustrated in Procedure~\ref{countarray}.

By changing, the condition $i \ge \mathcal{I}_{\mathcal{O}}(\mathcal{R}[index-1])$ to $i > \mathcal{I}_{\mathcal{O}}(\mathcal{R}[index-1])$, we can get combinations which contain no repeated elements. The recursion tree obtained when we generate all combinations of $\mathcal{L} = \{a,b,c\}$ with $r=2$ and $\mathcal{O}=\{a,b,c\}$ is shown in Figure~\ref{fig:rectreecomb}.

\begin{algorithm}
\begin{algorithmic}[1]
\caption{: APR3($index$) - Combinations}
\label{proc:apr3}
\STATE \textbf{Local:} $index$
\STATE \textbf{Global:} $CountArray$, $\mathcal{O}$, $r$, $p$ and $\mathcal{R}$
\medskip
\IF{$index > r$}
	\STATE Print $\mathcal{R}[1, \ldots, r]$ \label{proc:apr3:line:print}
	\RETURN
\ELSE
	\FOR{$i \gets 1$ to $p$} \label{proc:apr3:line:loop}
		\IF{$CountArray[i] \ge 1 \,\textbf{and}\,(index = 1$ \textbf{or} $i \ge \mathcal{I}_{\mathcal{O}}(\mathcal{R}[index-1]))$} \label{proc:apr3:line:cond}
			\STATE $\mathcal{R}[index] \gets o_i$ 
			\STATE $CountArray[i] \gets CountArray[i] - 1$
			\STATE APR3($index + 1$)
			\STATE $CountArray[i] \gets CountArray[i] + 1$
		\ENDIF
	\ENDFOR
\ENDIF
\end{algorithmic}
\end{algorithm}

\begin{figure*}[!htp]
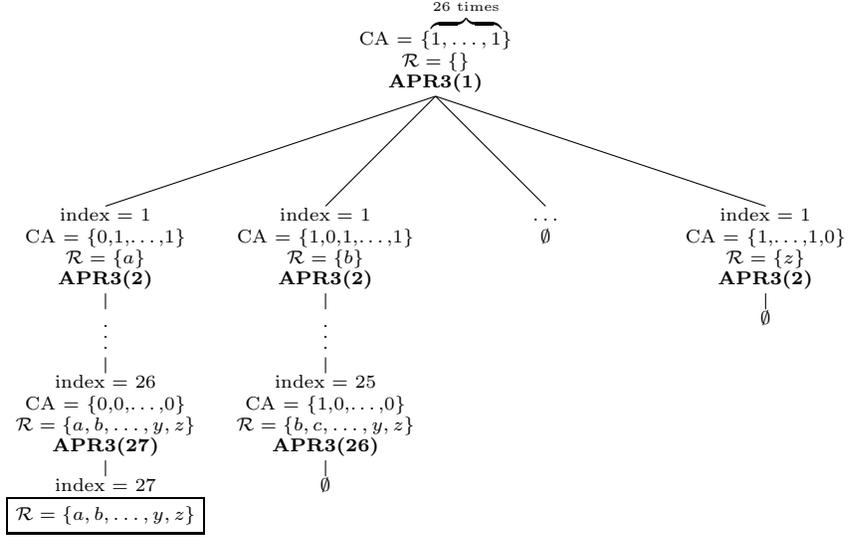
 \centering
{\scriptsize \Tree 
[
	.{CA = $\{\overbrace{1,\ldots,1}^{26 \text{ times}}\}$ \\ $\mathcal{R}=\{\}$ \\ \textbf{APR3(1)}} 
	[
		.{index = 1 \\ CA = \{0,1,\ldots,1\} \\ $\mathcal{R}=\{a\}$ \\ \textbf{APR3(2)}} 
		[
			.{\vdots}
			[
				.{index = 26 \\ CA = \{0,0,\ldots,0\} \\ $\mathcal{R}=\{a,b,\ldots,y,z\}$ \\ \textbf{APR3(27)}} 
				[.{index = 27 \\ \fbox{$\mathcal{R}=\{a,b,\ldots,y,z\}$}} ] 
			]
		]
	]
	[
		.{index = 1 \\ CA = \{1,0,1,\ldots,1\} \\ $\mathcal{R}=\{b\}$ \\ \textbf{APR3(2)}}
		[
			.{\vdots}
			[
				.{index = 25 \\ CA = \{1,0,\ldots,0\} \\ $\mathcal{R}=\{b,c,\ldots,y,z\}$ \\ \textbf{APR3(26)}} 
				[.{$\emptyset$} ] 
			]
		]
	]
	[
		.{$\ldots$ \\ $\emptyset$}
	]
	[
		.{index = 1 \\ CA = \{1,\ldots,1,0\} \\ $\mathcal{R}=\{z\}$ \\ \textbf{APR3(2)}}
		[
			.{$\emptyset$}
		]
	]
]}
\caption{Recursion Tree of APR3 for $\mathcal{L}=\{a,b,c,\ldots,z\}$; $r = 26$; $\mathcal{O}=\{a,b,\ldots,z\}$}
\label{fig:rectree_waste}
\end{figure*}

As is evident from Figure~\ref{fig:rectreecomb}, there can be many wasteful branches with certain kinds of input to APR3. For example, if we were to have $\mathcal{L} = \{a,b,\ldots,y,z\}$, $r=26$ and $\mathcal{O} = \{a,b,\ldots,y,z\}$, we would have a recursion tree similar to Figure~\ref{fig:rectree_waste}, where branches ending with $\emptyset$ denote wasteful branches, i.e. recursion branches that do not produce any output. All unnecessary computations are because we loop from $o_1 \ldots o_p$ at every index. However, based on two simple observations, we can completely eliminate all wasteful branches
\begin{enumerate}
 \item Once $o_x$ is assigned at a particular index, subsequent indices can only contain one of $o_{x+1} \ldots o_p$
 \item In the case where we have already assigned elements to the first $\beta$ indices of $\mathcal{R}$, and $\mathcal{R}[\beta] = o_y$, we need not loop if we are sure that there are fewer than $(r - \beta)$ elements left over, i.e., if $\sum_{i=y}^p CountArray[i] < (r - \beta)$, we can terminate looping in the current recursion level and return to the previous recursion level
\end{enumerate}

\begin{figure*}[!htp]
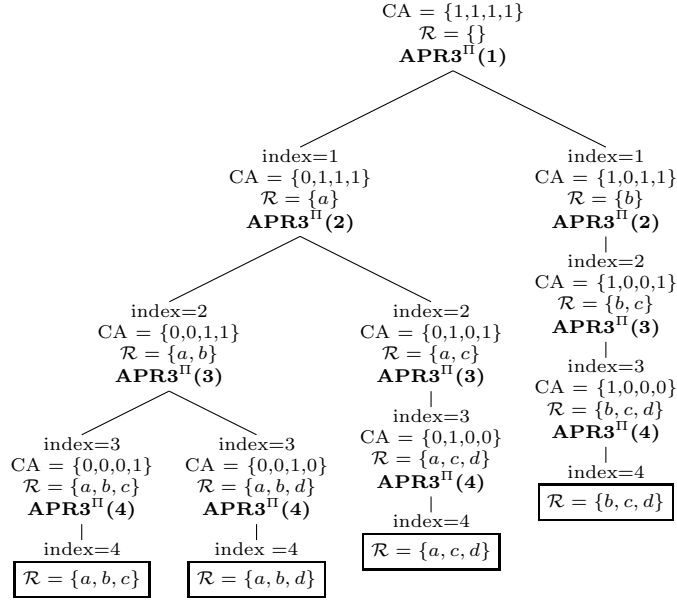
 \centering
{\scriptsize \Tree 
[.{CA = \{1,1,1,1\} \\ $\mathcal{R}=\{\}$ \\ \textbf{APR3$^{\Pi}$(1)}} 
	[.{index=1 \\ CA = \{0,1,1,1\} \\ $\mathcal{R}=\{a\}$ \\ \textbf{APR3$^{\Pi}$(2)}} 
		[.{index=2 \\ CA = \{0,0,1,1\} \\ $\mathcal{R}=\{a,b\}$ \\ \textbf{APR3$^{\Pi}$(3)}}
			[.{index=3 \\ CA = \{0,0,0,1\} \\ $\mathcal{R}=\{a,b,c\}$ \\ \textbf{APR3$^{\Pi}$(4)}}
				[.{index=4 \\ \fbox{$\mathcal{R}=\{a,b,c\}$}} ] ] 
			[.{index=3 \\ CA = \{0,0,1,0\} \\ $\mathcal{R}=\{a,b,d\}$ \\ \textbf{APR3$^{\Pi}$(4)}} 
				[.{index =4 \\ \fbox{$\mathcal{R}=\{a,b,d\}$}} ] ] ]
		[.{index=2 \\ CA = \{0,1,0,1\} \\ $\mathcal{R}=\{a,c\}$ \\ \textbf{APR3$^{\Pi}$(3)}}
			[.{index=3 \\ CA = \{0,1,0,0\} \\ $\mathcal{R}=\{a,c,d\}$ \\ \textbf{APR3$^{\Pi}$(4)}}
				[.{index=4 \\ \fbox{$\mathcal{R}=\{a,c,d\}$}} ] ] ] ]
	 [.{index=1 \\ CA = \{1,0,1,1\} \\ $\mathcal{R}=\{b\}$ \\ \textbf{APR3$^{\Pi}$(2)}}
		[.{index=2 \\ CA = \{1,0,0,1\} \\ $\mathcal{R}=\{b,c\}$ \\ \textbf{APR3$^{\Pi}$(3)}}
			[.{index=3 \\ CA = \{1,0,0,0\} \\ $\mathcal{R}=\{b,c,d\}$ \\ \textbf{APR3$^{\Pi}$(4)}}
				[.{index=4 \\ \fbox{$\mathcal{R}=\{b,c,d\}$}} ] ] ] ] ]}
\caption{Recursion Tree of APR3$^{\Pi}$ for $\mathcal{L}=\{a,b,c,d\}$; $r=3$; $\mathcal{O} = \{a,b,c,d\}$}
\label{fig:combCum}
\end{figure*}

\begin{algorithm}
\begin{algorithmic}[1]
\caption{: Building $CumCA$}
\label{proc:cumarray}
\STATE $CumCA[1] \gets 0$
\FOR{$i \gets 1$ to $p$}
	\STATE $CumCA[i+1] \gets CumCA[i] + CountArray[i]$
\ENDFOR
\end{algorithmic}
\end{algorithm}

\begin{algorithm}
\begin{algorithmic}[1]
\caption{: APR3$^{\Pi}$($index$) - Combinations \emph{Efficiently}}
\label{proc:apr3'}
\STATE \textbf{Local:} $index$
\STATE \textbf{Global:} $CumCA$, $CountArray$, $\mathcal{O}$, $r$, $p$ and $\mathcal{R}$
\medskip
\IF{$index > r$}
	\STATE Print $\mathcal{R}[1, \ldots, r]$
	\RETURN
\ELSE
	\STATE $i \gets \mathcal{I}_{\mathcal{O}}(\mathcal{R}[index-1])$
	\WHILE{$(CumCA[p+1] - CumCA[i+1] + CountArray[i]) > (r - index)$}
		\IF{$CountArray[i] \ge 1$}
			\STATE $\mathcal{R}[index] \gets o_i$ 
			\STATE $CountArray[i] \gets CountArray[i] - 1$
			\STATE APR3$^{\Pi}$($index + 1$)
			\STATE $CountArray[i] \gets CountArray[i] + 1$
		\ENDIF
		\STATE $i \gets i + 1$
	\ENDWHILE
\ENDIF
\end{algorithmic}
\end{algorithm}

This means that in Line~\ref{proc:apr3:line:loop} of Procedure~\ref{proc:apr3}, the loop limits would need to be made tighter. For simplicity, we modify the definition of $\mathcal{I}_{\mathcal{O}}(element)$ to return 0 if $element \notin \mathcal{O}$. Also, we assign to $\mathcal{R}[0]$ an element that does not exist in $\mathcal{O}$. This would mean that $\mathcal{I}_{\mathcal{O}}(\mathcal{R}[0]) = 0$. Now, Point 1 indicates that we can begin looping from $\mathcal{I}_{\mathcal{O}}(\mathcal{R}[index-1])$ instead of 1.

To incorporate Point 2, we build a \emph{Cumulative Array}~\cite{pearls} on $CountArray$. It is called $CumCA$ and is of size $(p+1)$. It is built such that $(CumCA[y+1] - CumCA[x]) = \sum_{i=x}^y CountArray[i]$. A quick implementation is shown Procedure~\ref{proc:cumarray}.

It is to be noted that $CumCA$ reflects $CountArray$ values at the beginning. We know that $CountArray$ values are constantly changing. Thus we can use $CumCA$ only for indices where $CountArray$ values are known to have not changed. Procedure~\ref{proc:apr3'} is a pseudocode representation of the modified version of APR3. This is denoted by APR3$^{\Pi}$. The effectiveness of APR3$^{\Pi}$ is made apparent in Figure~\ref{fig:combCum}. Analogous to the argument in Section~\ref{sec:anal:subsec:runtime}, the running time of APR3$^{\Pi}$ is in $O((n-r+1) \times r \times |\mathcal{C}^{r}(\mathcal{L})|)$

\subsection{Catalan Families}

In this section, we show how the algorithm can be modified to output all possible valid parenthesizations of $(n+1)$ factors. A valid parenthesization of $(n+1)$ factors can be defined as a sequence of $n$ opening brackets and $n$ closing brackets with the condition that at no point in the sequence should the number of closing brackets be greater than the number of opening brackets. We shall refer to the set of all possible valid parenthesizations of $(n+1)$ factors as $n$-Parenthesizations. Some previous work can be found in \cite{catalan2, knuth_perm}.

\begin{algorithm}
\caption{: Building $CountArray$}
\label{countCat}
\begin{algorithmic}[1]
\STATE \textbf{Input: $n$}
\STATE \textbf{Output: $CountArray$}
\medskip
 \STATE $CountArray[1] \gets n$
 \STATE $CountArray[2] \gets n$
\end{algorithmic}
\end{algorithm}

\begin{algorithm}
\begin{algorithmic}[1]
\caption{: APR4($index$) - $n$-Parenthesizations}
\label{permCat}
\STATE \textbf{Local:} $index$
\STATE \textbf{Global:} $CountArray$, $n$, $\mathcal{O}$ and $\mathcal{R}$
\medskip

\IF{$index > 2n$}
	\STATE Print $\mathcal{R}[1, \ldots, 2n]$
	\RETURN
\ELSE
\FOR{$i=1 \textbf{ to } 2$}
	\IF{($i \neq 2$ \textbf{ or } $CountArray[2]$ $>$ $CountArray[1]$) \textbf{and} $CountArray[i] \ge 1$}
		\STATE $\mathcal{R}[index] \gets o_i$ 
		\STATE $CountArray[i] \gets CountArray[i] - 1$
		\STATE APR4($index + 1$)
		\STATE $CountArray[i] \gets CountArray[i] + 1$
	\ENDIF
\ENDFOR
\ENDIF
\end{algorithmic}
\end{algorithm}

The input to the algorithm is $n$. We initialize $CountArray$ such that $CountArray[1]$ = $CountArray[2]$ = $n$. This can be viewed as an input list of $n$ opening brackets and $n$ closing brackets ($p = 2$). We then generate all possible permutations of this input list using $CountArray$ and make sure that at no point in a permutation we assign more closing brackets than opening brackets. We set $o_1 = "("$ and $o_2 = ")"$.

Once again, we see that we generate a different combinatorial structure using pretty much the same algorithm. The building of $CountArray$ is shown in Procedure~\ref{countCat} and the recursive generation of parenthesizations is shown in Procedure~\ref{permCat}. The recursive procedure is called APR4.

\subsection{Subsets}

\begin{algorithm}
\begin{algorithmic}[1]
\caption{: APR5($index$) - Subsets}
\label{proc:apr5}
\STATE \textbf{Local:} $index$
\STATE \textbf{Global:} $CountArray$, $\mathcal{O}$, $n$, $p$ and $\mathcal{R}$
\medskip
\STATE Print $\mathcal{R}[1, \ldots, index]$ \label{proc:apr5:line:print}
\IF{$index > n$}
	\RETURN
\ELSE
	\FOR{$i \gets 1$ to $p$}
		\IF{$CountArray[i] \ge 1 \,\textbf{and}\,(index = 1$ \textbf{or} $i > \mathcal{I}_{\mathcal{O}}(\mathcal{R}[index-1]))$} \label{proc:apr5:line:cond}
			\STATE $\mathcal{R}[index] \gets o_i$ 
			\STATE $CountArray[i] \gets CountArray[i] - 1$
			\STATE APR5($index + 1$)
			\STATE $CountArray[i] \gets CountArray[i] + 1$
		\ENDIF
	\ENDFOR
\ENDIF
\end{algorithmic}
\end{algorithm}

True to the formula $2^n = \sum_{i=0}^n \binom{n}{i}$, we know that all subsets of a list can be generated using an algorithm that generates all $r$-combinations. This is achieved by calling the $r$-combinations generating algorithm successively with: $0 \le r \le n$. However, the structure of APR allows us to do it more efficiently.

All subsets of a list of $n$ elements can be obtained by calling APR5 (shown in Procedure~\ref{proc:apr5}). It is a modified version of APR3 (shown in Procedure~\ref{proc:apr3}) with $r=n$ and the \textbf{Print} statement before the \textbf{if} block. $CountArray$ is again built exactly as was shown in Procedure~\ref{countarray}.

\subsection{Integer Compositions and Partitions}

A composition of an integer $n$ is a set of strictly positive integers which sum up to $n$~\cite{compositions}. For example, 3 has four compositions - $\big\{\{1,1,1\},\{1,2\},\{2,1\},\{3\}\big\}$. Generally, a composition of $n$ can contain any number from 1 to $n$. However, it is also interesting to study a variation of the problem wherein we are not allowed to use all numbers.

\begin{algorithm}
\begin{algorithmic}[1]
\caption{: APR6($index$) - Integer Compositions}
\label{proc:compo}
\STATE \textbf{Local:} $index$
\STATE \textbf{Global:} $n$, $\mathcal{O}$, $p$ and $\mathcal{R}$
\medskip
\IF{$n = 0$}
	\STATE Print $\mathcal{R}[1, \ldots, (index-1)]$
	\RETURN
\ELSE
	\FOR{$i \gets 1$ to $p$}
		\IF{$n \ge o_i$}
			\STATE $\mathcal{R}[index] \gets o_i$ 
			\STATE $n \gets n - o_i$
			\STATE APR6($index + 1$)
			\STATE $n \gets n + o_i$
		\ENDIF
	\ENDFOR
\ENDIF
\end{algorithmic}
\end{algorithm}

We now present an algorithm (APR6), based on the same algorithmic structure used so far, that given an $n$ and $\mathcal{O}=\{o_1,\ldots,o_p\}$, generates all possible compositions of $n$ using elements in $\mathcal{O}$. It is shown in Procedure~\ref{proc:compo}. An example recursion tree is shown in Figure~\ref{fig:compositions}. It is easy to modify Procedure~\ref{proc:compo} to generate all integer partitions, instead of integer compositions (analogous to how we generated combinations by using a permutations generating algorithm).

It is possible that when given an $\mathcal{O}=\{o_1,\ldots,o_p\}$, we have a limited number of some or all of the numbers. For example, we could have $n=15$; $\mathcal{O}=\{1,\ldots,15\}$ with the condition that no number can be used more than twice. This can be handled by maintaining an auxilary array $CountArray$ parallel to $\mathcal{O}$.

\begin{figure*}[!htp]
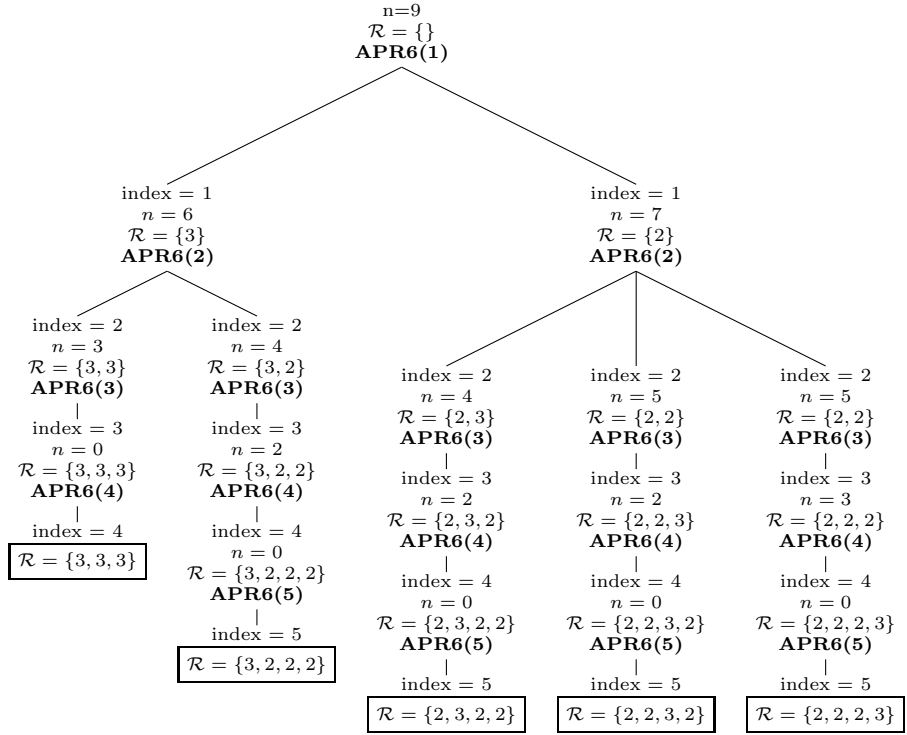
 \centering
{\scriptsize \Tree 
[
	.{n=9 \\ $\mathcal{R}=\{\}$ \\ \textbf{APR6(1)}} 
	[
		.{index = 1 \\ $n=6$ \\ $\mathcal{R}=\{3\}$ \\ \textbf{APR6(2)}} 
		[
			.{index = 2 \\ $n=3$ \\ $\mathcal{R}=\{3,3\}$ \\ \textbf{APR6(3)}}
			[
				.{index = 3 \\ $n=0$ \\ $\mathcal{R}=\{3,3,3\}$ \\ \textbf{APR6(4)}} 
				[
					.{index = 4 \\ \fbox{$\mathcal{R}=\{3,3,3\}$}} 
				] 
			]
		]
		[
			.{index = 2 \\ $n=4$ \\ $\mathcal{R}=\{3,2\}$ \\ \textbf{APR6(3)}}
			[
				.{index = 3 \\ $n=2$ \\ $\mathcal{R}=\{3,2,2\}$ \\ \textbf{APR6(4)}} 
				[
					.{index = 4 \\ $n=0$ \\ $\mathcal{R}=\{3,2,2,2\}$ \\ \textbf{APR6(5)}}
					[
						.{index = 5 \\ \fbox{$\mathcal{R}=\{3,2,2,2\}$}}
					]
				] 
			]
		]
	]
	[
		.{index = 1 \\ $n=7$ \\ $\mathcal{R}=\{2\}$ \\ \textbf{APR6(2)}}
		[
			.{index = 2 \\ $n=4$ \\ $\mathcal{R}=\{2,3\}$ \\ \textbf{APR6(3)}}
			[
				.{index = 3 \\ $n=2$ \\ $\mathcal{R}=\{2,3,2\}$ \\ \textbf{APR6(4)}} 
				[
					.{index = 4 \\ $n=0$ \\ $\mathcal{R}=\{2,3,2,2\}$ \\ \textbf{APR6(5)}}
					[
						.{index = 5 \\ \fbox{$\mathcal{R}=\{2,3,2,2\}$}}
					]
				] 
			]
		]
		[
			.{index = 2 \\ $n=5$ \\ $\mathcal{R}=\{2,2\}$ \\ \textbf{APR6(3)}}
			[
				.{index = 3 \\ $n=2$ \\ $\mathcal{R}=\{2,2,3\}$ \\ \textbf{APR6(4)}} 
				[
					.{index = 4 \\ $n=0$ \\ $\mathcal{R}=\{2,2,3,2\}$ \\ \textbf{APR6(5)}}
					[
						.{index = 5 \\ \fbox{$\mathcal{R}=\{2,2,3,2\}$}}
					]
				] 
			]
		]
		[
			.{index = 2 \\ $n=5$ \\ $\mathcal{R}=\{2,2\}$ \\ \textbf{APR6(3)}}
			[
				.{index = 3 \\ $n=3$ \\ $\mathcal{R}=\{2,2,2\}$ \\ \textbf{APR6(4)}} 
				[
					.{index = 4 \\ $n=0$ \\ $\mathcal{R}=\{2,2,2,3\}$ \\ \textbf{APR6(5)}}
					[
						.{index = 5 \\ \fbox{$\mathcal{R}=\{2,2,2,3\}$}}
					]
				] 
			]
		]
	]
]}
\caption{Recursion Tree of APR6 for $n=9$; $\mathcal{O} = \{3,2\}$}
\label{fig:compositions}
\end{figure*}

\section{Conclusion}
\label{sec:conc}

Algorithms to generate combinatorial structures will always be needed, simply because of the fundamental nature of the problem. One could need different kinds of combinatorial structures for different kinds of input. Thus, having one common effective algorithm, as the one proposed in this paper, which solves many problems, would be useful.

One must note that a few other non-trivial adaptations - generating all palindromes of a string, generating all positive solutions to a diophantine equation with positive co-efficients, etc, are also possible.

\section{Future Research}
\label{sec:future}

The proposed algorithm conclusively solves quite a few fundamental problems in combinatorics. Further research directed towards achieving even more combinatorial structures while preserving the essence of the proposed algorithm should assume topmost priority.

One could also perform some probabilistic analysis and derive tighter average case bounds on the runtime of the algorithms.

\section*{Acknowledgement}

We would like to thank Prof. K V Vishwanatha from the CSE Dept. of R V College of Engineering, Bangalore for his comments.

\bibliographystyle{amsplain}
\bibliography{APR}
\label{lastpage}

\end{document}